\documentclass[pdflatex,sn-mathphys-num]{sn-jnl}

\usepackage{graphicx}%
\usepackage{multirow}%
\usepackage{amsmath,amssymb,amsfonts}%
\usepackage{amsthm}%
\usepackage{mathrsfs}%
\usepackage[title]{appendix}%
\usepackage{xcolor}%
\usepackage{textcomp}%
\usepackage{manyfoot}%
\usepackage{booktabs}%
\usepackage{algorithm}%
\usepackage{algorithmicx}%
\usepackage{algpseudocode}%
\usepackage{listings}%

\theoremstyle{thmstyleone}%
\newtheorem{theorem}{Theorem}
\theoremstyle{thmstyletwo}%

\theoremstyle{thmstylethree}%
\newtheorem{definition}{Definition}%

\theoremstyle{plain}
\newtheorem{corollary}{Corollary}

\theoremstyle{plain}
\newtheorem{lemma}{Lemma}

\theoremstyle{remark}
\newtheorem{observation}{Observation}

\usepackage{booktabs}
\usepackage{tabularx}
\usepackage{url}  
\usepackage{array}
\usepackage{hyperref} 
\usepackage{booktabs}

\usepackage{graphicx}

\usepackage{adjustbox}

\newcommand{\hdr}[1]{\multicolumn{1}{c}{\fontsize{4.5pt}{3pt}\selectfont\textbf{#1}}}
\newcommand{\hdrc}[1]{\multicolumn{1}{c}{\fontsize{4.5pt}{4pt}\selectfont\textbf{#1}}}

\begin{document}

\title[Article Title]{
Ultrabubble enumeration via a lowest common ancestor approach

}


\author*[1]{\fnm{Athanasios} \sur{E. Zisis}}\email{athanas.zisis@gmail.com}

\author*[1,2,3]{\fnm{Pål} \sur{Sætrom}}\email{pal.satrom@ntnu.no}

\affil*[1]{\orgdiv{Department Of Computer Science}, \orgname{Norwegian University of Science and Technology}, \orgaddress{\street{Sem Sælands vei 9}, \city{Trondheim}, \postcode{7034}, \country{Norway}}}

\affil[2]{\orgdiv{Department of Clinical and Molecular Medicine}, \orgname{Norwegian University of Science and Techonology}, \orgaddress{\street{Erling Skjalgsons gate 1}, \city{Trondheim}, \postcode{7491}, \country{Norway}}}

\affil[3]{\orgdiv{Sentral Stab}, \orgname{St. Olavs Hospital HF}, \orgaddress{\city{Trondheim}, \postcode{7006}, \country{Norway}}}

\abstract{
Pangenomics uses graph-based models to represent and study the genetic variation between individuals of the same species or between different species. In such variation graphs, a path through the graph represents one individual genome. Subgraphs that encode locally distinct paths are therefore genomic regions with distinct genetic variation and detecting such subgraphs is integral for studying genetic variation.

Biedged graphs is a type of variation graph that use two types of edges, black and grey, to represent genomic sequences and adjacencies between sequences, respectively. Ultrabubbles in biedged graphs are minimal subgraphs that represent a finite set of sequence variants that all start and end with two distinct sequences; that is, ultrabubbles are acyclic and all paths in an ultrabubble enter and exit through two distinct black edges. Ultrabubbles are therefore a special case of snarls, which are minimal subgraphs that are connected with two black edges to the rest of the graph.

Here, we show that any bidirected graph can be transformed to a bipartite biedged graph in which lowest common ancestor queries can determine whether a snarl is an ultrabubble. 
This leads to an $O(Kn)$ algorithm for finding all ultrabubbles in a set of $K$ snarls, improving on the prior naive approach of $O(K(n+m))$ in a biedged graph with $n$ nodes and $m$ edges. Accordingly, our benchmark experiments on real and synthetic variation graphs show improved run times on graphs with few cycles and dead end paths, and dense graphs with many edges.
}




\maketitle

\section{Introduction}\label{sec2}

Instead of relying on a single reference, pangenomics studies all genomes from the species together. This approach helps reveal both the common and variable parts of the genomes and reduces reference bias. It is used to build shared genomic references, align sequencing reads across different haplotypes, and compare genomes on a large scale \cite{6}.

A pangenome encompasses the complete set of genetic material within a species, including both core genes found in all individuals and the variable accessory genes that differ across populations. This comprehensive view of a species' genetic diversity allows a deeper understanding of how different genetic elements contribute to phenotypic traits, adaptation, and evolutionary processes. By studying pangenomes, researchers can examine not only the shared genetic features but also the variations that play a role in individual and population-level differences \cite{6}.

A variation graph is a bidirected DNA sequence graph (see Preliminaries) that provides a structured way to represent pangenomes, capturing genetic variations within a single unified framework. Nodes represent DNA sequences and include a unique identifier along with an implicit reverse complement that indicates the opposite orientation of the sequence. Each orientation of a node is referred to as a node strand. Edges indicate the adjacencies that exist between the sequences of the nodes they join, defining the graph’s base topology by specifying how nodes are linked. Paths in the graph represent sequences or walks over the node strands, effectively encoding specific sequences, like genomes and haplotypes embedded in the graph. 
A biedged graph is an alternative, but equivalent representation of the bidirected variation graph \cite{1}. In the biedged graph, the DNA sequences are encoded along special black edges while the adjacencies are encoded as grey edges.
Using such graphs as references enables efficient modeling of genomic variation across populations and reduces reference bias compared with a single linear reference genome \cite{8}. 

In variation graphs, some of the subgraphs that represent allelic diversity are bubbles, superbubbles, and ultrabubbles \cite{1,2,3}, which all are variants of snarls (see Preliminaries) \cite{1}. A bubble in a directed graph is a subgraph $X$ that has two alternative, node-disjoint paths between the same entry $s$ and exit $t$ nodes in the graph \cite{3}. A superbubble generalizes bubbles by allowing multiple but finite internal paths between $s$ and $t$, while requiring $X$ to be minimal and $s$ and $t$ to be the only entry and exit nodes of $X$  \cite{3}. 
Superbubbles are relevant because they represent distinct variant subgraphs that have low computational cost and are likely to be error-free as they have finite paths and no dead ends \cite{4,9,16}.

BubbleGun \cite{3}, by implementing Ondera’s Algorithm \cite{10}, enumerates superbubbles and thus also bubbles, by scanning for $(s,t)$ candidates and 
aborting the scan if the subgraph contains cycles or nodes (other than $t$) without outgoing edges. 
Their algorithm has average linear time complexity and runs nearly linearly in practice for large graphs, but has a quadratic worst case of $O(n(n+m))$, where $n$ is the number of nodes and $m$ is the number of edges \cite{10}.

Ultrabubbles generalize superbubbles in biedged graphs by defining the ultrabubble subgraph $X$ to be connected by two distinct black edges (2-BEC) to the rest of the graph  \cite{1}. Similar to superbubbles, an ultrabubble $X$ is minimal and acyclic, and has only single entry and exit points $s$ and $t$. Snarls further generalize ultrabubbles by only requiring the subgraph to be 2-BEC and minimal. 

Snarls can be computed by first building a cactus graph \cite{1,17} that summarizes the graph’s min-cut structure and then producing a tree-based nesting of snarl components. This approach runs efficiently (near-linear in practice), in sharp contrast to a naive brute-force search over all side pairs. 
The cactus-based hierarchy returns a non-overlapping decomposition set of snarls, thereby omitting overlapping snarls that are still valid according to the formal snarl definition \cite{1}.

A snarl can be checked if it is an ultrabubble by running a breadth first search (BFS) or depth first search (DFS) in its subgraph to check for cycles or nodes (other than $s$ or $t$) with no outgoing grey edges. Therefore, assuming that the number of snarls in a biedged graph is $K$, the complexity of enumerating all ultrabubbles is $O(K(n+m))$. 

This paper focuses on identifying all ultrabubbles in variation graphs with a given set of snarls. Given the set of $K$ snarls of a variation graph with $n$ nodes and $m$ edges, we give an $O(Kn)$ algorithm that enumerates all the ultrabubbles of the graph. The algorithm, which is based on lowest common ancestor queries, improves on the $O(K(n+m))$ complexity of the naive approach and its run-time is close to $O(K)$ in most practical graphs.

The remainder of the paper is organized as follows. Section 2 sets up the preliminaries, Section 3 develops two algorithmic families: Section 3.1, which is based on classical methods, and Section 3.2, which presents our theoretical framework with full proofs and complexity bounds.
Section 4 details Materials and Methods, including datasets, preprocessing, and implementation specifics. Section 5 reports our results and compares our approach with the existing naive approach and relates the findings to the theory. Section 6 has discussions and conclusions, and Section 7 provides a supplement of specific theory mentioned in earlier Sections.

\section{Preliminaries}\label{sec2}
A directed graph $G = (V_{G}, E_{G})$ consists of a set of nodes $V_{G}$ connected by a set of directed edges $E_{G}$, where each edge $e = (v_i, v_j) \in E_{G}| v_i, v_j \in V_{G}$ represents a relationship with a specified direction (i.e. a path from node $v_i$ to node $v_j$). A bidirected graph generalizes the concept of directed graphs by allowing each end of an edge to have its own orientation ("left" or "right"), representing the side of the node to which the edge is connected. 

Replacing each node $v$ in a bidirected graph $D = (V_D, E_D)$ with two nodes $v_{\text{left}}, v_{\text{right}}$ connected by a special black edge transforms the bidirected graph into an equivalent biedged graph $B(D) = (V_{B(D)}, E_{B(D)})$ \cite{1}. The original edges of $D, E_D$ are referred to as grey edges. Consequently, $E_{B(D)} = B_{B(D)} \cup G_{B(D)}$, where $G_{B(D)} = E_D$ and $B_{B(D)} = \{\{v_{\text{left}}, v_{\text{right}}\} |
v\in V_D \}$. In the following, for any node $x$ in a biedged graph, $x'$ refers to the node on the other side of $x$'s connected black edge. When orientation matters, $L$ and $R$ refer to $v_{\text{left}}$ and $v_{\text{right}}$, respectively.

A cycle in a graph refers to a closed path that begins and ends at the same node, with 
no repeated nodes or edges along the way \cite{7}. In a directed graph, a directed cycle is a closed directed path that starts and ends at the same node without repeating any other node. Moreover, if self-loops are allowed then the $(v,v)$ arc is a length 1 directed cycle \cite{11}.

A connected component (component for short) is a subgraph of a graph in which every pair of nodes is connected by a path, with no connection to other nodes outside the subgraph. 
A bridge edge in a graph is an edge whose removal causes the graph to become disconnected, splitting it into separate components. This edge is essential for maintaining the graph’s overall connectivity \cite{7}.

In a directed graph $G = (V_{G}, E_{G})$, a strongly connected component is a maximal subgraph $G' = (V_{G'}, E_{G'})$ where $V_{G'}\in V_G, E_{G'}\in E_G$ such that every pair of nodes in $G'$ is connected by a path. That is, for each node pair in $G'$ there should be a path in both directions and $G'$ cannot be extended with additional nodes or edges from $G$ without breaking this property. 

In a directed graph, a node without incoming edges is called a source while a node without outgoing edges is called a sink \cite{11}. Both sources and sinks are single node strongly connected components.

A snarl \cite{1} is a minimal subgraph of a biedged graph, identified by a pair of distinct 
nodes $\{x, y\}$ with distinct black edges (i.e. $x$ and $y$ are not connected by a black edge), such that $\{x, y\}$ have the following properties:
\begin{itemize}
  \item Separation: removing the black edges connected to $x$ and $y$ divides the graph into
disconnected components. One component $X$ contains $x$ and $y$ but not $x’$ and $y’$.
  \item Minimality: there is no single black edge within $X$ that, if removed, would split $X$ into two or more disconnected components.
\end{itemize}

In a directed graph $G$, a subgraph defined by a pair of distinct nodes $(x, y)$ is called a superbubble if it meets the following criteria \cite{10}:
\begin{itemize}
    \item Reachable: there is a path from $x$ to $y$.
    \item Matching: the set of nodes $X$ reachable from $x$ without traversing $y$ matches the
set of nodes from which $y$ is reachable without traversing $x$; here, traversing a node means entering and exiting the node.
    \item Acyclic: the subgraph formed by $X$ has no cycles.
    \item Minimal: no other node in $X$ other than $y$ can form a valid superbubble pair with $x$; the same applies symmetrically for $y$.
\end{itemize}
In short, superbubbles are minimal subgraphs with single entry and exit nodes (i.e. there are no other sources or sinks in the subgraph) and a finite number of paths between those nodes.

Ultrabubbles generalize superbubbles in biedged graphs. Specifically, an ultrabubble is a snarl that is both acyclic and tip free, where a tip is a node that has no incoming or outgoing (that is, no incident) grey edges \cite{1}.

\section{Algorithms}\label{sec2}
In the biedged graph representation of variation graphs, the black edges represent sequence blocks and grey edges represent their adjacencies. Specific sequences, such as individual genomes, are encoded as paths through the graph. 
Any such path through the graph always alternates between black and grey edges and, in general, every orientation and traversal of both black and grey edges is valid and thus single segment node inversions and longer paths of them can take place \cite{1,2}.

In the following, we assume that every biedged graph is directed; 
that is, all edges can only be traversed in one direction, preventing inversions, including long inversions spanning multiple nodes. Specifically, 
we define the following traversal rules:
\begin{itemize}
    \item a black edge can be traversed only from its left ($L$) node to its right ($R$) node; 
    \item a grey edge can be traversed only an $R$ node of a black edge to a $L$ node of a black edge.
\end{itemize}
Through this approach, which is referred to as alignment-rescued inversions \cite{12}, we achieve a simpler graph topology. 
Specifically, the biedged graph is bipartite, consisting of the node sets $L$ and $R$.
From now on, when referring to traversals of the graph, we always mean that we follow the above traversal rules.

In practice, variation graphs, such as those that are available in GFA format, are bidirected, but these can be converted in linear time to a directed graph with a simple, only forward-flow form. Every sequence is kept in forward direction; whenever a connection would point the other way, we add a reverse complemented copy of that sequence and connect it normally with the corresponding orientation of the grey edge (adjacencies) and then we convert them to biedged graphs that have the characteristics we need; see Section 7 for details. 

Note that by following this restricted approach of the general biedged graphs we do not lose information since the above transformation preserves all paths of the given bidirected graph, in GFA format, since every path of it has an equivalent forward path in the directed one and you can always map results back to the original structure. Therefore, this approach stays solid while giving a simpler topology of the graph, at the expense of doubling the number of edges and nodes in the graph, in the worst case.

Without loss of generality, we assume that all biedged graphs are connected, since for the non-connected ones, each connected component can be processed independently as a single connected biedged graph. We assume that each biedged graph has a unique source (hereafter "root") such that all other nodes in the graph can be reached by traversing the graph from the root. 
As the root has no incoming edges, it represents the start of the pangenomic sequences encoded by the biedged graph. In practice, variation graphs may not have a unique root, but this can be resolved by transforming the graph as described in Section 7.
Finally, we assume that a biedged graph has at least one sink ("end"). 
If there are multiple such nodes, we choose the sink at the deepest level of a BFS tree started at the root; if there are more than one such node, we choose one of those arbitrarily. If the graph contains no sink, one can be constructed as described in Section 7.

We define the components $C_{\text{start}}$ and $C_{\text{end}}$ as any subgraph of a biedged graph that includes the root and the end nodes, respectively. Clearly, 
$C_{\text{start}}$=$C_{\text{end}}$ is possible.  
We define the following left-right orientation in a biedged graph. A node $x$ is to the left of node $y$ 
if and only if the shortest path from the root to $x$ is shorter than the shortest path from the root to $y$ in the biedged graph. Here, the length of a path is the number of edges traversed along the path. 

Considering the left-right orientation and a snarl $\{x,y\}$, we call the leftmost and rightmost frontier nodes $sn_1$ and $sn_2$, respectively. In case both nodes have the same shortest path from the root, we arbitrarily assign $sn_1$ and $sn_2$. 

Because of the previously defined traversal rules for biedged graphs, in a BFS from the 
root, every $R$ (right) node of a black edge can only be reached through the corresponding $L$ (left) node of the same black edge, since incoming grey edges at an $R$ node are forbidden. Moreover, since the graphs we are dealing with are connected and every node is reached from the root, 
every black edge will be an edge of the BFS tree.

A root in a biedged graph with more than one node is, by our previous definitions, always only connected to a single black edge. Consequently, the root is a left node and the connected node is a right node $R_r$ with distance 1 to the root. 
Moreover, since by the traversal rules and the structure of the graph, black and grey edges alternate, every $L$ node and $R$ node have even and odd distances from the root, respectively. Corollary~\ref{cor:2} follows from these observations. 

\begin{corollary}\label{cor:2}
A snarl defined by an $R$ and $L$ node cannot have frontiers $sn_1$ and $sn_2$ with equal distances from the root and thus we can always sort-characterize correctly and uniquely such a snarl as $R-L$ or $L-R$ depending on which of the two frontier nodes are closer to the root.
\end{corollary}

As every acyclic directed graph has a source and a sink \cite{11}, we get Corollary~\ref{cor:2.1}.
\begin{corollary}\label{cor:2.1}
    A biedged graph (isomorphic to a directed graph) that has no cycles has at least one source and one sink.
\end{corollary}

Following the snarl definition \cite{1}, the minimal snarl is a single grey edge that connects the two nodes of its adjacent black edges. Note that whereas this minimal snarl does not actually represent variations in the graph, it is still a valid snarl according to the definition and is also a valid minimal (trivial) ultabubble.

\begin{definition}\label{2.2}
    We call a minimal snarl a trivial snarl (see Fig.~\ref{fig:snarls}\textbf{D}).
\end{definition}

\subsection{Ultrabubbles and snarls by the naive approach}\label{sec:UBnaive}
Given the set of the snarls of a biedged graph as a set of tuples, where each tuple has the frontier nodes $sn_1$ and $sn_2$ of the snarl, to find which of them are ultrabubbles, we need for each snarl to check the subgraph defined by its frontier nodes $sn_1$ and $sn_2$ for the presence of cycles and tips, for example by using a DFS or BFS. If the snarl is free of cycles and tips it is an ultrabubble; otherwise, it is not.
 Therefore, assuming that the number of snarls is bounded by $K$, then the complexity of the naive approach is $O(K(n+m))$ since we need to run a DFS or BFS for each snarl.

\subsection{Ultrabubbles and snarls by lowest common ancestor queries}\label{sec:UBlca}
In a biedged graph, the edges of the source and sink are inherently black edges; consequently, the source and the sink are left and right nodes of the leftmost and the rightmost black edge, respectively. Moreover, by definition, the source cannot have an outgoing grey edge and the sink cannot have an incoming grey edge.Therefore, both the root and sink form tips, which gives Lemma~\ref{lem:3.1} and Theorem~\ref{thm:3.2}.

\begin{figure}[t]
  \centering
  \includegraphics[width=\linewidth]{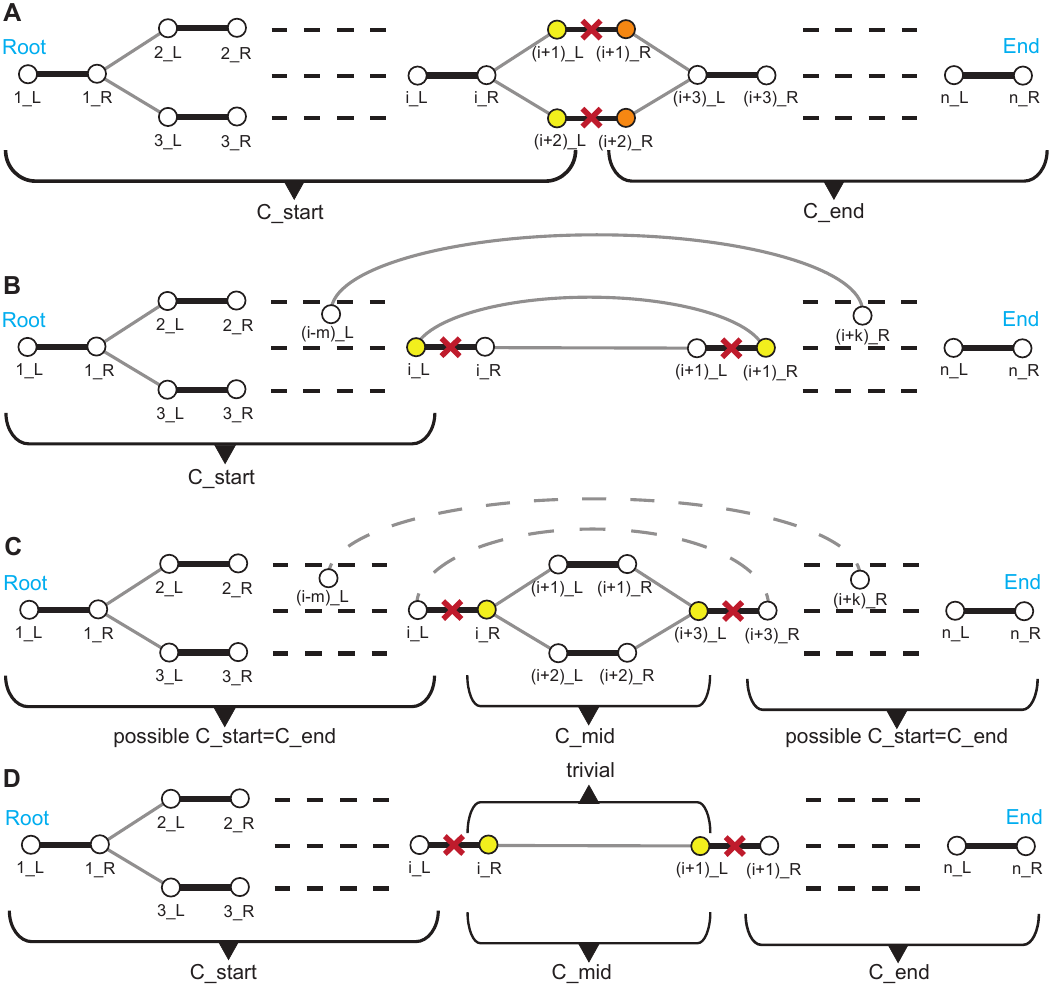} 
  \caption{Snarls that (\textbf{A-B}) cannot and (\textbf{C-D}) can form ultrabubbles in biedged graphs. Snarls are colored yellow or orange; red crosses indicate the corresponding removed black edges. \textbf{(A)} Any $L-L$ (left-left) or $R-R$ (right-right) snarl cannot form an ultrabubble because it belongs to the component $C_{\text{start}}$ or $C_{\text{end}}$, respectively.
  \textbf{(B)} Any $L-R$ snarl cannot form an ultrabubble  because the $L$ node, which is part of the snarl, will be in the $C_{\text{start}}$ component.
  \textbf{(C)}  
  Any $R-L$ snarl partitions the graph into three or two components. The former happens only if there are no grey edges connecting $C_{\text{start}}$ and $C_{\text{end}}$, in which case the removed black edges are bridge edges. In the latter case, at least one such grey edge exist and then $C_{\text{start}}$=$C_{\text{end}}$.
  \textbf{(D)} An $R-L$ snarl connected by a single grey edge is both a trivial snarl and a trivial ultrabubble.
}
  \label{fig:snarls}
\end{figure}

\begin{lemma}\label{lem:3.1}
Any snarl whose subgraph contains either $C_{\text{start}}$ or $C_{\text{end}}$, cannot form an ultrabubble.
\end{lemma}

\begin{theorem}\label{thm:3.2}
A snarl that is defined by 2 left (or 2 right) nodes of black edges can never form an ultrabubble (see Fig.~\ref{fig:snarls}\textbf{A}). 
\end{theorem}
\begin{proof}
Assume that we have a snarl defined by 2 left nodes $x$, $y$. If this snarl contains $C_{\text{start}}$ or $C_{\text{end}}$, then by Lemma~\ref{lem:3.1}, the snarl cannot form an ultrabubble. Assume that such a snarl is an ultrabubble and, consequently, it does not contain either $C_{\text{start}}$ or $C_{\text{end}}$. Then this snarl will belong to a different component, $C_{\text{mid}}$. Since it is an ultrabubble, $C_{\text{mid}}$ is an acyclic biedged graph. But then by Corollary~\ref{cor:2.1} , $C_{\text{mid}}$ should have at least one source and one sink. Since the nodes of the snarl are both left nodes of black edges, they cannot have outgoing edges, as each of the outgoing black edges from $x$ and $y$ have been removed, and, by definition, there can be no outgoing grey edges from left nodes of black edges. Consequently, both $x$ and $y$ constitute sinks in $C_{\text{mid}}$ and therefore $C_{\text{mid}}$ should have at least one source $s$ other than the nodes $x$, $y$ that define the snarl. The source $s$ cannot have incoming grey edges because this would imply either the presence of a cycle (contradiction) or that $s$ is not a source node because there is another node $s'$ forming the source in a subgraph consisting of at least $s'$, a black edge, another node, and the grey edge that is incoming to $s$. Note that we have not deleted any other black edge apart from the two edges incident of the snarl's frontier nodes. Hence, the source node is a left node of a black edge without any loops ending up to it, since the snarl has no cycles. But then this node, as explained above, is not adjacent to any grey edges and thus it is a tip which contradicts the assumption that the snarl is an ultrabubble. In the case that we have 2 right nodes defining a snarl, the proof is the same since these nodes will constitute sources and thus, we will have the presence of at least one sink which will form again a tip.
\end{proof}

\begin{observation}\label{obs:3.3}
By the definition of a snarl $\{x,y\}$, the deletion of the 2 black edges $(x,x’)$ and $(y,y’)$ can transform the biedged graph from consisting of one connected component to consisting of either 2 or 3 connected components, see Fig.~\ref{fig:snarls}\textbf{C}. 
The case with 3 components occurs if and only if the 2 black edges of the nodes that define the snarl are bridge edges in the biedged graph. In the other case of 2 components, the snarl could include either (i) $C_{\text{start}}$, (ii) $C_{\text{end}}$, (iii) both, or (iv) neither of $C_{\text{start}}$ and $C_{\text{end}}$. Therefore, for a snarl to contain neither $C_{\text{start}}$ nor $C_{\text{end}}$, either the 2 black edges $(x,x’)$ and $(y,y’)$ are bridge edges (3 components case) or the snarl is nested inside an area of the graph that does not include either $C_{\text{start}}$ or $C_{\text{end}}$; see Fig.~\ref{fig:snarls}\textbf{C}. 
\end{observation}

\begin{lemma}\label{lem:3.4}
    Any snarl where its leftmost frontier $sn_1$ is a left node $(L)$ of a black edge, cannot form an ultrabubble (see Fig.~\ref{fig:snarls}\textbf{B}). 
\end{lemma}

\begin{proof}
    The $sn_1$ node is, by definition, the closest to the root of the two frontier nodes of the snarl. Moreover, whereas $sn_1$, since it is a left node, can have one or more incoming gray edges, its opposite node $sn_1^{\prime}$, at the other end of the black edge connected to $sn_1$, cannot have any incoming grey edges as backtracking along black edges is, by definition, not allowed. Consequently, the subgraph formed after deleting the black edge, $(sn_1$, $sn_1^{\prime})$, and that contains node $sn_1$ will also contain the root; that is, the root and $sn_1$ belong to the same connected component, $C_{\text{start}}$. Therefore, by Lemma~\ref{lem:3.1}, such a snarl 
    cannot form an ultrabubble.
\end{proof}

\begin{figure}[H]
  \centering
  \includegraphics[width=\linewidth]{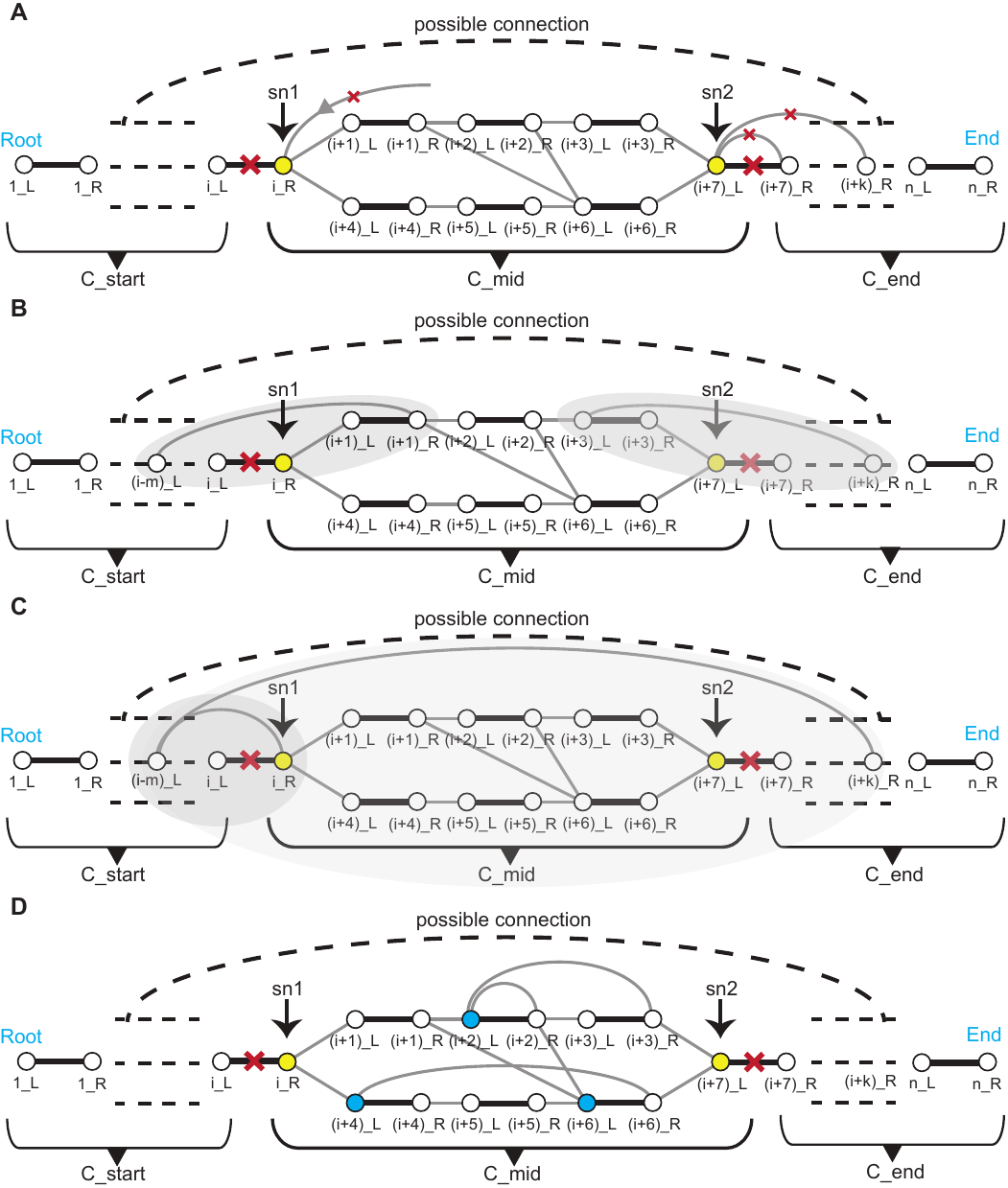} 
 \caption{Cycles that are (\textbf{A-C}) incompatible and (\textbf{D}) compatible with snarls. 
 \textbf{(A)} The frontier nodes $sn_1$, $sn_2$ of an $R-L$ snarl cannot be cycle closing nodes. By definition, $sn_1$ cannot have incoming grey edges. For $sn_2$, an incoming cycle closing grey edge would violate the snarl definition, since it would place $sn_2$ and $sn_2'$ (here, $(i+7)\_L$ and $(i+7)\_R$) in the same component.  
 \textbf{(B-C)} An internal node of an $R-L$ snarl cannot \textbf{(B)} be a cycle closing node of a cycle that belongs partially inside and outside the snarl, or \textbf{(C)} have a grey edge connected to a node outside of the snarl, as this would violate the separation criterion of the snarl definition. 
 \textbf{(D)} Cycle closing nodes in an $R-L$ snarl that exist between the frontier nodes $sn_1$ and $sn_2$ of the snarl represent a cycle nested in the snarl's frontiers. Note that 
 the identified cycle closing nodes will depend on the DFS execution; here, either $(i+4)\_L$ or $(i+6)\_L$ represent the same nested cycle.
 }
  \label{fig:cyclenodes}
  \end{figure}

By Theorem~\ref{thm:3.2} and Lemma~\ref{lem:3.4}, we can directly exclude all $R-R$, $L-L$, and $L-R$ snarls as ultrabubbles.
Consequently, we will only check the ultrabubble criteria for $R-L$ snarls. 

Assume that we have a rooted biedged graph as defined previously with uniquely labeled nodes. By starting in the root and running a modified DFS for directed graphs that follow our traversal rules (see Section 7), 
we find in linear time all tips and all cycle closing nodes of every cycle in the graph, which we mark blue.
Note that because of the traversal rules, only left nodes can be cycle closing nodes.

\begin{theorem}\label{thm:3.5}
An $R-L$ snarl $\{x,y\}$, where $sn_1=x$ is a right node and $sn_2=y$ is a left node, cannot have $x$ or $y$ as cycle closing nodes (see Fig.~\ref{fig:cyclenodes}\textbf{A}).
\end{theorem}
\begin{proof}
    First, consider the rightmost frontier $sn_2 = y$.
Assume that $y$ has a cycle closing edge. This means that there is an incoming grey edge from a node $c$ that is further away from the root than $y$ and thus there must exist a path from $y$ to $c$ (otherwise there would be no cycle). Starting with $y$, that path must start with the black edge $(y,y’)$, since there is no other way to leave $y$ because $y$ is a left node and left nodes do not have outgoing grey edges. Therefore, there exists a path from $y’$ to $c$ and by removing the black edge $(y,y’)$, as per the snarl definition, the cycle will be disrupted, but there will still be a grey edge from $c$ to $y$. Consequently, after removing the black edge $yy'$, $y’$, $c$, and $y$ are nodes in the same component,
which violates the separation criterion for snarls.
Second, in an $R-L$ snarl, the leftmost frontier $x$ cannot be a cycle closing node since it is a right node of a black edge, which according to the definition cannot have incoming grey edges. 
\end{proof}

By Theorem~\ref{thm:3.5}, we have Lemma~\ref{lem:3.6}.

\begin{lemma}\label{lem:3.6}
If a cycle closing node is present in an $R-L$ snarl, 
then its corresponding cycle is completely nested in the snarl’s subgraph (see Fig.~\ref{fig:cyclenodes}\textbf{B,C}).     
\end{lemma}
\begin{proof}
    Let $c$ be a cycle closing node between the frontier nodes $x$,$y$ of an $R-L$ snarl. Assume that the corresponding cycle is not completely nested in the subgraph of the snarl. Then a path of that cycle has to be outside of the snarl subgraph, either to the snarl's right or left. Since the snarl is connected to the rest of the graph through $x$ and $y$, 
    there is at least one node $d$ that participates in this cycle with $d=y’$ or $d=x'$, both cases being possible. But then there will be a path that connects $d$ with $c$. Consequently, if $d$ lies on the right or the left of the snarl, then $y’$, $d$, $c$, and $y$ or $d$, $x’$, $c$, and $x$, respectively, will be in the same component. 
    Both cases violate the separation criterion of the snarl definition.  
\end{proof}

\begin{lemma}\label{lem:3.7}
    A tip that is in an $R-L$ snarl’s subgraph can be any node of the subgraph except a descendant of $sn_2$.
\end{lemma}
\begin{proof}
    In an $R-L$ snarl $\{sn_1, sn_2\}$, $sn_2$ can only have incoming grey edges since it is a left node. 
    Assume one of those incoming grey edges is on a path originating from a tip 
    – that is, a node with no incident grey edges. Then this tip and the corresponding path to $sn_2$ is not traversable from the root as  we demand, since this would imply an outgoing grey edge from a left node ($sn_2$), and thus our assumptions about the root are violated (see Fig.~\ref{fig:dagroot}\textbf{A}). If such a tip is present, we have Observation~\ref{obs:3.7.1} below.
\end{proof}

\begin{observation}\label{obs:3.7.1}
    A tip connected with a grey edge to a left node is a potential root for the adapted sequence orientation of the biedged graph. 
    Consequently, the graph has more than one potential root (see Fig.~\ref{fig:dagroot}\textbf{B}; we discuss this real case scenario in Sections 4 and 7.) It is obvious that $sn_1$, since it is a right node, can have an outgoing grey edge 
    that leads to a tip, which would be a sink in the biedged graph. 
\end{observation}

An ultrabubble is an $R-L$ snarl $\{sn_1, sn_2\}$ that contains no cycles or tips. By Lemma~\ref{lem:3.6}, a cycle is present in an $R-L$ snarl subgraph if and only if a cycle closing node lies between $sn_1$ and $sn_2$. Moreover, by Lemma~\ref{lem:3.7}, a tip in an $R-L$ snarl  lies between $sn_1$ and $sn_2$, since it can be a descendant of $sn_1$ but not a descendant of $sn_2$. Therefore, to check if an $R-L$ snarl is an ultrabubble it is sufficient to check if there are tips or cycle closing nodes that lie between $sn_1$ and $sn_2$. Thus, both tips and cycle closing nodes will be checked for the same property, so we place both in a set labeled ftip. 
Consequently, 
an $R-L$ snarl is not an ultrabubble if and only if at least one node $t\in \text{ftip}$ lies between $sn_1$ and $sn_2$ in the subgraph of the snarl. In other words, $t$ is a descendant of $sn_1$ but not a descendant of $sn_2$ (see Fig.~\ref{fig:ftip}\textbf{B}).

Given the above conclusions, our approach to deal with the ultrabubble classification of a given $R-L$ snarl is based on the lowest common ancestor (LCA). 
Specifically, a node $t$ is located within an $R-L$ snarl $\{sn_1, sn_2\}$ if $LCA(t, sn_1) = sn_1$ and $LCA(t, sn_2) \ne sn_2$, where $LCA(x, y)$ is the LCA of $x$ and $y$
(see Fig.~\ref{fig:ftip}).
Using LCA to test whether a snarl is an ultrabubble therefore involves $O(n)$ LCA queries in worst case.

\begin{figure}[H]
  \centering
  \includegraphics[width=\linewidth]{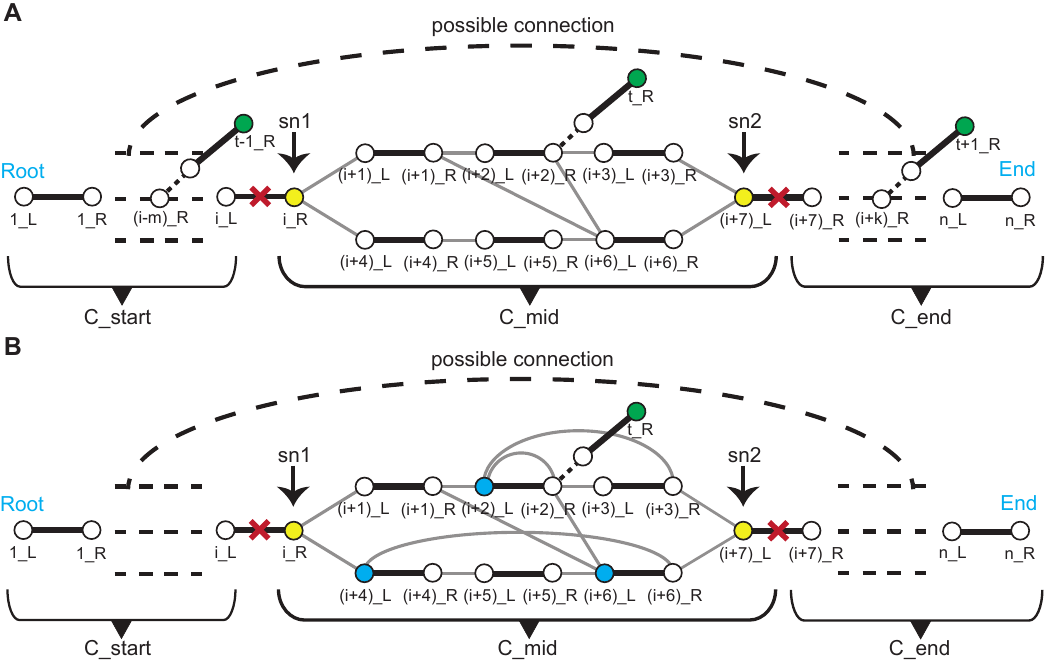} 
  \caption{Lowest common ancestor (LCA) queries identify whether a node is located within an $R-L$ snarl (yellow).
  \textbf{(A)} For the three tips (green) $(t-1)\_R$, $t\_R$, and $(t+1)\_R$ located, respectively, to the left, within, and to the right of the snarl, the following hold:  
  $LCA((t-1)\_R, sn_1)\neq sn_1$, $LCA(t\_R, sn_1)=sn_1$, $LCA(t\_R, sn_2)\neq sn_2$, and $LCA((t+1)\_R, sn_2)= sn_2$.
  \textbf{(B)} The same set of LCA checks apply to both tips (green) and cycle closing nodes (blue). 
 }
  \label{fig:ftip}
 \end{figure}

Whereas the LCA is well defined in directed acyclic graphs (DAG) \cite{28}, a given biedged graph is not necessarily a DAG. 
In general directed graphs, the LCA queries can be answered efficiently by first contracting the graph’s strongly connected components (SCC) to build a DAG called the condensation DAG. The queries are then answered on this DAG, and each result can be mapped back to a node of the original graph using a fixed representative rule; see Section 7.2 for details.
However, as we are only interested in determining if a given node $t$ is located between the frontiers $sn_1$ and $sn_2$ of a given $R-L$ snarl, we can instead run the LCA queries in a BFS tree from the root of the biedged graph (Theorem~\ref{thm:3.8}).

\begin{theorem}\label{thm:3.8}
Given a biedged graph $B$ with a unique root and corresponding BFS tree $B_t$, an $R-L$ snarl $\{sn_1, sn_2\}$, and a node $t$ located within the snarl, the LCA queries between $t$ and $sn_1$ and $t$ and $sn_2$ in $B_t$ will answer if $t$ is between $sn_1$ and $sn_2$ and not more left and outside the frontier than $sn_1$ (see Fig.~\ref{fig:lca}\textbf{B}).
\end{theorem}
\begin{proof}
    Every $R-L$ oriented snarl in a biedged graph $B$ clearly has a unique entry node and a unique exit node (frontier nodes $sn_1$ and $sn_2$, respectively) because of the definition and the traversal rules.     
    Assume that the LCA query $LCA(sn_1, t) = c_1$
    in the BFS tree $B_t$ of the graph,
    and that $c_1$ is located more to the left than $sn_1$.
    As $t$ is a descendant of $sn_1$ in the snarl's subgraph, this assumption means that there are two distinct paths $P_1$ and $P_2$ from $c_1$ to $sn_1$ and $t$, respectively, and that $sn_1$ is not part of $P_2$.
    $P_1$ must include $sn_1'$, as the only path to $sn_1$ is through the black edge $(sn_1', sn_1)$ (by the traversal rules and Theorem~\ref{thm:3.5}). 
    Moreover, $P_2$ cannot include $(sn_2, sn_2')$, as this edge cannot be used to enter the snarl $\{sn_1, sn_2\}$ as it has the opposite direction.
    Consequently, $P_2$ must include some other edge that connects $t$ to $c_1$, which means that $sn_1'$ and $sn_1$ are still part of the same component, violating the separation criterion of the snarl definition. 
    The same logic holds for LCA queries between $t$ and $sn_2$, as both are descendants of $sn_1$; that is, for $LCA(sn_2, t) = c_2$, $c_2$ must be part of the snarl's subgraph, unless there is some other edge beside $(sn_1', sn_1)$ and $(sn_2, sn_2')$ connecting $\{sn_1, sn_2\}$ to the rest of $B$ (and thereby $c_2$ to $t$), thereby violating the separation criterion.
\end{proof}
\begin{figure}[H]
  \centering
  \includegraphics[width=\linewidth]{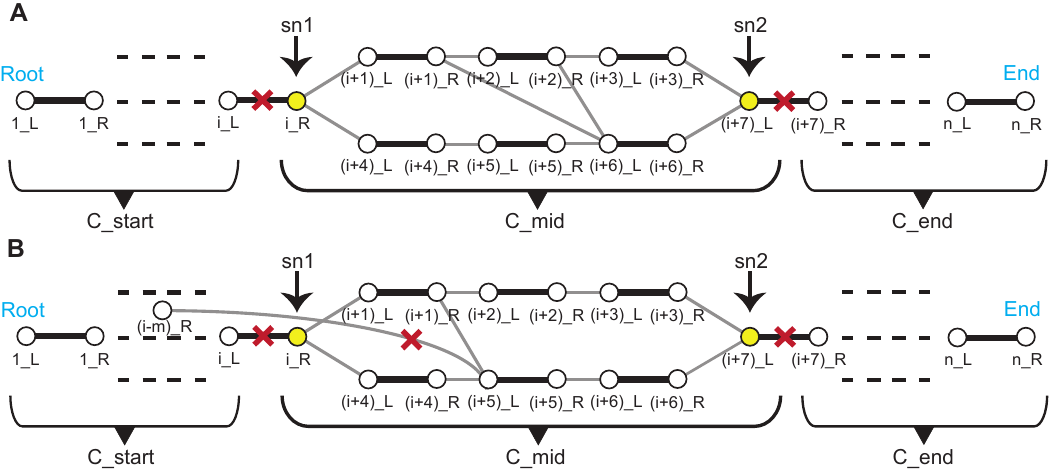} 
  \caption{
  Lowest common ancestor (LCA) queries in a breadth first search (BFS) tree ($LCA_{B_t}$) identify whether a node is located within an $R-L$ snarl (yellow).
  \textbf{(A)} 
  $LCA_{B_t}$ can give a result that is closer to the root than that of the LCA in the biedged graph ($LCA_B$); for example, $LCA_{B_t}((i+3)\_L, (i+6)\_L)=(i+1)\_R$, whereas $LCA_B((i+3)\_L, (i+6)\_L)=(i+2)\_R$.  
  \textbf{(B)} $LCA_{B_t}$ can never give an answer higher than $sn_1$, as such an answer would violate the separation criterion of the snarl.
  }
  \label{fig:lca}
\end{figure}

By Theorem~\ref{thm:3.8}, we have Corollary~\ref{cor:3.9}.

\begin{corollary}\label{cor:3.9}
For any $R-L$ snarl $\{sn_1, sn_2\}$ of a given biedged graph $B$, it is sufficient to use LCA queries in the BFS tree $B_t$ of the graph to determine if a node in the $ftip$ set is present within $\{sn_1, sn_2\}$, despite that LCA queries in the BFS tree $B_t$ may give different answers than LCA queries in $B$. 
\end{corollary}

From Corollary~\ref{cor:3.9}, we get Algorithm~\ref{alg:algo1}, which, given a biedged graph $B$ that satisfies the traversal rules, the set $S_{RL}$ of $R-L$ snarls in $B$, the set $ftip$ of tips and cycle closing nodes in $B$, and the BFS tree $B_t$ of $B$, returns the set $U$ of ultrabubbles in $B$. As LCA queries in trees can be answered in constant time, the complexity of this algorithm is $O(K\cdot |ftip|)$, which is $O(Kn)$ in worst case.

\begin{algorithm}[t]
\caption{Enumerating ultrabubbles in a biedged graph $B$}\label{alg:algo1}
\begin{algorithmic}[1]
\Require The set $S_{RL}$ of $R-L$ snarls in $B$, the set $ftip$, and the BFS tree $B_t$ of $B$.
\Require $LCA_{B_t}(a,b)$: Lowest common ancestor of $a$ and $b$ in $B_t$.

\State $U \gets \emptyset$
\For{each $\{sn_1,sn_2\} \in S_{RL}$}
    \State $\mathit{F} \gets \textbf{true}$
    \For{each $t \in ftip$}
        \If{$LCA_{B_t}(t,sn_1)=sn_1$ and $LCA_{B_t}(t,sn_2)\neq sn_2$}
            \State $\mathit{F} \gets \textbf{false}$
            \State \textbf{break}
        \EndIf
    \EndFor
    \If{$\mathit{F} = \textbf{true}$}
        \State $U \gets U \cup \{(sn_1,sn_2)\}$
    \EndIf
\EndFor
\State \textbf{return} $U$
\end{algorithmic}
\end{algorithm}

\section{Materials and methods}\label{sec3}

\subsection{Graphs}
We downloaded bidirected variation graphs in GFA format from available sources (Table~\ref{tab:one}) and used a custom Python script to construct synthetic biedged graphs with different sizes, bubble complexities, and numbers of cycles and tips (Tables~\ref{tab:two}-\ref{tab:three}).
As some of the variation graphs were disconnected, we isolated and worked with a subset of their components or a subgraph of them; for example, lpa120.gfa is a 120-node subgraph of LPA.gfa. 

\begin{table}[h]
\caption{Datasets used in the experiments}\label{tab:one}
\small
\setlength{\tabcolsep}{4pt}
\renewcommand{\arraystretch}{1.15}
\begin{tabularx}{\textwidth}{@{}
  >{\raggedright\arraybackslash}p{0.30\textwidth}
  >{\raggedright\arraybackslash}X
  >{\raggedright\arraybackslash}p{0.18\textwidth}
@{}}
\toprule
Name & Source URL & Reference \\
\midrule
\path{C4-90.gfa} &
\url{https://zenodo.org/records/6617246} &
\cite{15} \\

\path{chr19.pan.fa.a2fb268.e820cd3.9ea71d8.smooth.gfa} &
\url{https://42basepairs.com/browse/s3/human-pangenomics/pangenomes/scratch/2021_07_30_pggb/chroms?file=chr19.pan.fa.a2fb268.e820cd3.9ea71d8.smooth.gfa.gz} &
\cite{14}, \cite{18} \\

\path{chr6.pan.fa.a2fb268.2ff309f.1300c8a.smooth.gfa} &
\url{https://s3-us-west-2.amazonaws.com/human-pangenomics/index.html?prefix=pangenomes/scratch/2021_05_06_pggb/gfas/} &
\cite{14}, \cite{18} \\

\path{chr6.C4.gfa} &
\url{https://github.com/pangenome/odgi/tree/master/test} &
\cite{13} \\

\path{chrM.pan.fa.6626ff2.7748b33.72587dd.smooth.gfa} &
\url{https://s3-us-west-2.amazonaws.com/human-pangenomics/index.html?prefix=pangenomes/scratch/2021_07_30_pggb/chroms/} &
\cite{14}, \cite{18} \\

\path{LPA.gfa} &
\url{https://github.com/pangenome/odgi/tree/master/test} &
\cite{13} \\

\path{lpa120.gfa} &
\url{https://github.com/pangenome/odgi/tree/master/test} &
\cite{13} \\

\path{MHC-57b.gfa} &
\url{https://zenodo.org/records/6617246} &
\cite{15} \\

\path{cerevisiae.fa.gz.d1a145e.417fcdf.7493449.smooth.final.gfa} &
\url{http://hypervolu.me/~erik/yeast/cerevisiae.fa.gz.d1a145e.417fcdf.7493449.smooth.final.gfa.gz} or \url{ https://doi.org/10.5281/zenodo.18484366 } &
\cite{18}, \cite{19} \\
\bottomrule
\end{tabularx}
\end{table}

\subsection{Graph preprocessing}

Given a bidirected graph $D$ in GFA format, we used the following procedure to transform it to a biedged graph with a structure in line with our traversal rules. First, we replaced each segment node in $D$ with a left and right node connected by a black edge $(L,R)$. Second, we deleted all edges connecting two right nodes or two left nodes. We used this simpler approach to handle graphs containing such $R-R$ and $L-L$ edges instead of transforming the graph as described in Section~\ref{sec:s1}, as we only used the graphs for benchmarking. The resulting graph was a biedged graph with a structure consistent with the previously defined traversal rules, and thereby isomorphic to a directed graph. Third, if the graph had multiples sources, we constructed a root as illustrated in Fig.~\ref{fig:msroots}. After these transformations, all the variation graphs used in our experiments had a root, but note that Section~\ref{sec:s2} describes a general approach of constructing a single source in a biedged graph isomorphic to a directed graph.

\subsection{Identifying snarls}

We used two approaches to identify snarls in the preprocessed graphs. First, we used the brute force naive approach, which for each pair of nodes of the graph, checks the snarl definition criteria of separation and minimality. This approach has cubic time complexity, but identifies all snarls, including nested snarls. Due to its time complexity, we used this naive approach on a subset of the graphs (Table~\ref{tab:two}).

Second, we used the vg tool's option of identifying snarls, which is based on Paten et al. \cite{1} and produces a hierarchical and nested structure of snarls. Specifically, the tool outputs the set of snarls as tuples of frontier nodes and gives additional info which of them are acyclic. We kept only the acyclic snarls for further processing, as these are the only possible ultrabubbles. As the algorithm used by the vg tool has close to linear time practical complexity, we used the vg tool on all the graphs.

\subsection{Identifying ultrabubbles}

We used the naive approach (Section~\ref{sec:UBnaive}) and Algorithm~\ref{alg:algo1} to identify the ultrabubbles in a set of snarls. The naive approach checked each snarl for presence of tips or cycles. For Algorithm~\ref{alg:algo1}, we first use the BFS tree to compute the distance from the root to the snarls' frontier nodes to identify the set of $R-L$ snarls as per Corollary~\ref{cor:2}. Second, in linear time we find all tips and cycle closing nodes with their corresponding closing grey edges as detailed in Section~\ref{sec:cycleclose}. Third, we pre-compute all paths from the root to all nodes by the Range Minimum Query (RMQ) – Euler tour combination technique, as explained in Section~\ref{sec:lcalinear}, thus allowing each LCA query to be answered in $O(1)$ time.

\subsection{Software}

We used Bandage [https://rrwick.github.io/Bandage/, ver. 0.8.1] \cite{27} for visualizing GFA files and vg [https://github.com/vgteam/vg,  ver. 1.62.0] \cite{8} to produce the set of nested snarls.

All Python scripts were run on a laptop with Intel\textsuperscript{\textregistered} Core\textsuperscript{\texttrademark} i7-1280p that contains 14 Cores and 20 Threads and has 16 GB RAM.

GPT (ChatGPT, developed by OpenAI) was used to support code development. 
The code is available at https://github.com/athanasios-zisis/UltraLCA.

\section{Results}\label{sec4}

We used a set of real variation graphs and synthetic biedged graphs to benchmark the LCA-based method (Algorithm~\ref{alg:algo1}) against the naive approach (Section~\ref{sec:UBnaive}) for identifying ultrabubbles in a set of snarls; see Methods for details on how the graphs were processed. We used the vg tool for identifying snarls, but for a set of small graphs, we also used the naive brute force approach for identifying snarls, as this method identifies all snarls according to their definition (Table~\ref{tab:two}). As expected, the naive approach identified a larger set of snarls, but both sets contained the same set of ultrabubbles (compare Tables~\ref{tab:two} and \ref{tab:three}).

\begin{table}[h]
\caption{Running times for the LCA and naive approach for enumerating all ultrabubbles in the snarls produced by the naive approach in different graphs. "NAME(.gfa)" is the file name of the graph; "NODES" and "EDGES" are the number of nodes and edges, respectively, in the graph; "REMOVED LL\&RR EDGES" is the number of $L-L$ and $R-R$ edges removed during preprocessing; "TIPS/CYCLES/FTIP" are the number of tips, cycle-closing nodes, and nodes in the $ftip$ set, respectively; "SNARLS NAIVE" is the number of snarls identified by the naive algorithm; "SN TIME" is the running time (in seconds) for identifying the snarls; "UL." is the number of ultrabubbles; "REJ" is the number of rejected $R-L$ snarls; "PRE TABLE" is the running time for computing the RMQ distance table; and "ALGO" and "NAIVE" are the running times for Algorithm~\ref{alg:algo1} and the naive approach, respectively.}\label{tab:two}
\tiny
\setlength{\tabcolsep}{1pt}
\renewcommand{\arraystretch}{0.95}

\begin{tabularx}{\textwidth}{@{}
  >{\raggedright\arraybackslash}X  
  *{11}{>{\centering\arraybackslash}c} 
@{}}


\midrule
\hdr{NAME(.gfa)} & \hdr{NODES} & \hdr{EDGES} & \hdr{\shortstack{REMOVED \\  LL\&RR \\ EDGES}} & \hdr{\shortstack{ TIPS/ \\
CYCLES/\\
FTIP}
} & \hdr{\shortstack{SNARLS\\
NAIVE}
} & \hdr{\shortstack{SN \\TIME}} & \hdr{UL.} & \hdr{REJ} & \hdr{\shortstack{PRE \\TABLE}} & \hdr{ALGO} & \hdr{NAIVE} \\
\midrule
\path{C4-90_VG_rem_nroot} & 19 & 20 & 10 & 7/0/7 & 29 & 0 & 7 & 22 & 0 & \textbf{0} & 0.001 \\
\path{lpa120.gfa} & 120 & 169 & - & 2/0/2 & 82 & 6 & 56 & 26 & 0 & \textbf{0} & 0.010 \\
\path{MHC-57b_VG} & 1068 & 1496 & - & 2/0/2 & 592 & 3417 & 376 & 216 & 0 & \textbf{0} & 0.59 \\
\path{chrM.pan.fa.6626ff2.7748b33.72587dd.smooth} & 1395 & 1887 & - & 2/0/2 & 1339 & 11577 & 456 & 883 & 0 & \textbf{0} & 2.85 \\
\path{chr6.pan.fa.a2fb268.2ff309f.1300c8a.smooth_comp6_rem_nroot} & 1700 & 2283 & 1 & 4/0/4 & 1643 & 25855 & 584 & 1059 & 0.1 & \textbf{0.01} & 9.188 \\
\path{chr6.C4.gfa} & 1748 & 2366 & - & 2/2/3 & 830 & 52070 & 582 & 248 & 0 & \textbf{0} & 0.81 \\
\midrule
\path{SYNAL1} & 236 & 289 & - & 6/4/9 & 225 & 23.9 & 138 & 87 & 0 & \textbf{0} & 0.058 \\
\path{SYNAL2} & 439 & 569 & - & 6/10/14 & 379 & 148 & 233 & 146 & 0 & \textbf{0.01} & 0.214 \\
\path{SYNAL3} & 598 & 767 & - & 6/12/16 & 511 & 373 & 337 & 174 & 0 & \textbf{0.01} & 0.33 \\
\path{SYNAL4} & 1107 & 1426 & - & 5/13/17 & 860 & 2376.3 & 596 & 264 & 0 & \textbf{0.01} & 0.81 \\

\bottomrule
\end{tabularx}
\end{table}

In most cases Algorithm~\ref{alg:algo1} outperforms the naive approach (Table~\ref{tab:three}). However, in cases where the ftip set is very large, Algorithm~\ref{alg:algo1} is slower, since we run an LCA query for every $t$ in the $ftip$ set. 

\begin{table}[h]
\caption{Running times for the LCA and naive approach for enumerating all ultrabubbles in the snarls produced by the vg tool in different graphs. Here, vg was run with the flag -T, which outputs all snarls that are nested and acyclic plus the trivial ones; "VG SNARLS -T" is the number of such snarls; see Table~\ref{tab:two} for details regarding the other columns. 
}\label{tab:three}
\tiny
\setlength{\tabcolsep}{1pt}
\renewcommand{\arraystretch}{0.95}

\begin{tabularx}{\textwidth}{@{}
  >{\raggedright\arraybackslash}X
  *{10}{>{\centering\arraybackslash}c}
@{}}

\toprule
\hdrc{NAME(.gfa)} & \hdrc{NODES} & \hdrc{EDGES}&  \hdrc{\shortstack{REMOVED \\  LL\&RR \\ EDGES}} & \hdrc{TIPS} & \hdrc{\shortstack{VG\\
SNARLS\\-T}
} & \hdrc{UL.}  & \hdrc{REJ} & \hdrc{\shortstack{PRE \\TABLE}} & \hdrc{ALGO} & \hdrc{NAIVE} \\
\midrule
\path{C4-90_VG_rem_nroot} & 19 & 20 & 10 & 7 & 9 & 7 & 2 & 0 & 0 & 0 \\
\path{lpa120.gfa} & 120 & 169 & - & 2 & 56 & 56 & 0 & 0 & \textbf{0} & 0.001 \\
\path{MHC-57b_VG} & 1068 & 1496 & - & 2 & 376 & 376 & 0 & 0 & \textbf{0} & 0.008 \\
\path{chrM.pan.fa.6626ff2.7748b33.72587dd.smooth} & 1395 & 1887 & - & 2 & 456 & 456 & 0 & 0 & \textbf{0} & 0.007 \\
\path{chr6.pan.fa.a2fb268.2ff309f.1300c8a.smooth_comp6_rem_nroot} & 1700 & 2283 & 1 & 4 & 586 & 584 & 2 & 0 & \textbf{0} & 0.008 \\
\path{chr6.C4.gfa} & 1748 & 2366 & - & 2 & 584 & 582 & 2 & 0 & \textbf{0} & 0.009 \\
\path{LPA.gfa} & 3751 & 5195 & - & 2 & 1305 & 1305 & 0 & 0 & \textbf{0.01} & 0.018 \\
\path{chr6.pan.fa.a2fb268.2ff309f.1300c8a.smooth_comp5_nroot} & 5475 & 7374 & - & 2 & 1896 & 1894 & 2 & 0 & \textbf{0.01} & 0.025 \\
\path{cerevisiae.fa.gz.d1a145e.417fcdf.7493449.smooth.final_comp11_rem_nroot} & 54721 & 73869 & 3 & 9 & 18172 & 18162 & 10 & 0.4 & \textbf{0.19} & 0.27 \\
\path{cerevisiae.fa.gz.d1a145e.417fcdf.7493449.smooth.final_comp4_nroot} & 72349 & 97855 & - & 7 & 23907 & 23898 & 9 & 0.6 & \textbf{0.27} & 0.343 \\
\path{cerevisiae.fa.gz.d1a145e.417fcdf.7493449.smooth.final_comp7_rem_nroot} & 216342 & 292679 & 7 & 55 & 71457 & 71395 & 62 & 2 & 3.41 & \textbf{1.172} \\
\path{chr19.pan.fa.a2fb268.e820cd3.9ea71d8.smooth_comp14_rem_nroot} & 2858719 & 3985788 & 786 & 458 & 904944 & 904804 & 140 & 35.8 & 325.85 & \textbf{17.608} \\
\path{chr6.pan.fa.a2fb268.2ff309f.1300c8a.smooth_comp1_rem_nroot} & 4906290 & 6856764 & 280 & 606 & 1555547 & 1555411 & 136 & 70.1 & 847.49 & \textbf{29.743} \\
\midrule
\path{Synth1} & 8091 & 10133 & - & 200 & 4662 & 4453 & 209 & 0 & 0.68 & \textbf{0.04} \\
\path{Synth2} & 173879 & 219006 & - & 152 & 102585 & 102421 & 164 & 1.5 & 12.11 & \textbf{0.978} \\
\path{Synth3} & 27886 & 35086 & - & 7 & 16542 & 16536 & 6 & 0.2 & \textbf{0.13} & 0.149 \\
\path{Synth4} & 37105 & 47488 & - & 7 & 5041 & 5035 & 6 & 0.3 & \textbf{0.08} & 0.14 \\
\path{Synth5} & 834549 & 1071462 & - & 157 & 467294 & 467110 & 184 & 8.7 & 58.37 & \textbf{5.485} \\
\path{Synth6} & 211938 & 273297 & - & 3 & 119602 & 119601 & 1 & 2 & \textbf{0.68} & 1.355 \\

\bottomrule
\end{tabularx}
\end{table}

\section{Discussion and conclusions}\label{sec2}
We have presented the theoretical framework showing that LCA queries can be used to determine if a snarl is an ultrabubble in a rooted bipartite biedged graph. This approach, which has complexity $O(Kn)$ improves on the $O(K(n+m)$ complexity of the naive approach, where $K$, $n$, and $m$ are the number of snarls, nodes, and edges, respectively. Theoretically, compared with the naive approach, this LCA approach should have improved run time for graphs with few tips and cycles or for dense graphs with many edges, which is supported by our benchmarking experiments on practical graphs. We note that additional run-time improvements may be had by integrating the LCA-based into existing approaches for enumerating snarls.

\section{Supplementary Material}\label{sec2}
\subsection{Building from a bidirected graph in gfa format, a biedged graph isomorphic to a directed graph}\label{sec:s1}
Given a bidirected graph in GFA format we can transform it into a biedged graph without inversions and/or backtracking as follows.
We scan the Graphical Fragment Assembly file (GFA) once. For each link, force both ends to forward orientation (``$+ +$''). If an end is reverse (``$-$''), redirect that end to a new node that stores the reverse complement of the same sequence segment creating that node the first time it is needed. If the link has an overlap string, reverse its operation order only when exactly one end flips; if both or neither flip, leave it unchanged. Update paths the same way by replacing any reverse step with the corresponding reverse complement node in forward orientation. Because every line is read and rewritten once and each sequence or overlap string is transformed at most once, the whole conversion complexity is linear time.

By the above process we end up to a biedged graph that is isomorphic to a directed graph as explained in this paper.

\subsection{Constructing a root (source) and an end (sink) of a variation-pangenome graph}\label{sec:s2}

\begin{figure}[H]
  \centering
  \includegraphics[width=\linewidth]{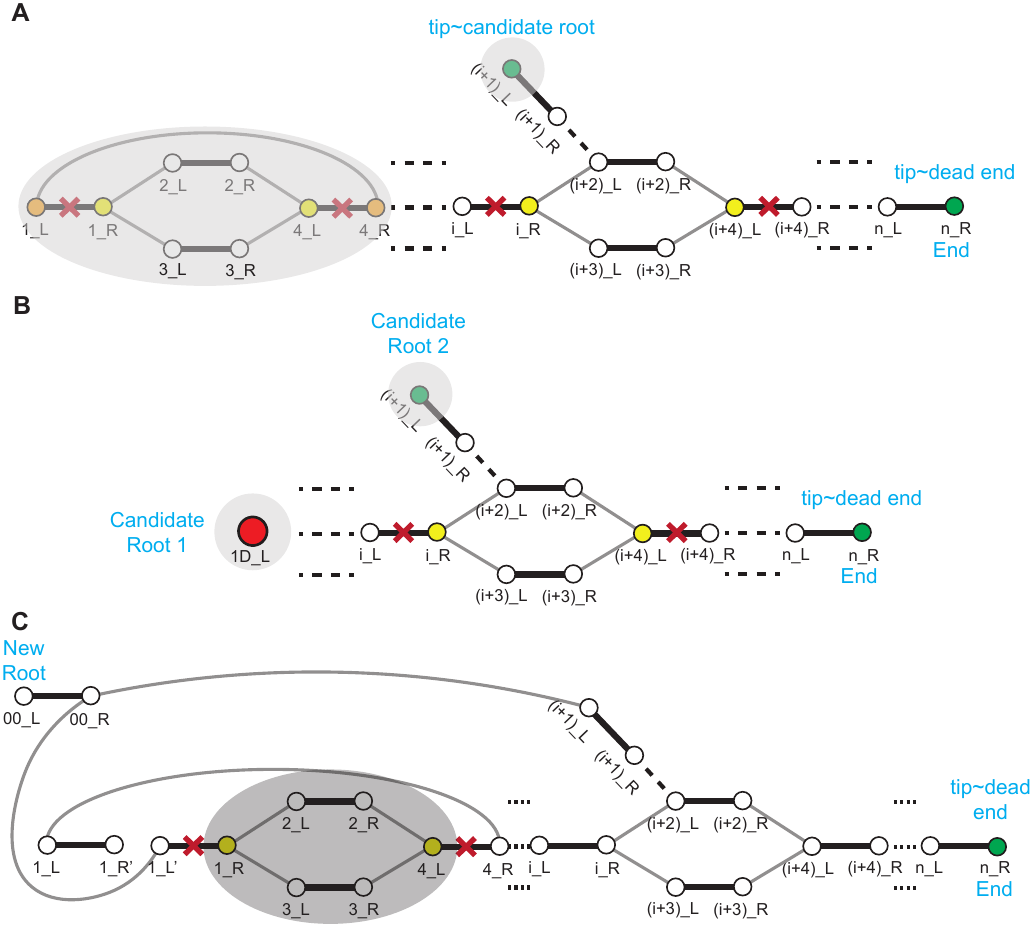} 
  \caption{
  \textbf{A)}A graph that has not a working root is illustrated , which means the presence of cycle structures or candidate roots that cannot reach all nodes of the graph, or a combination. Note that both $(i+1)\_L$ and $n\_R$ are tips but the first one is a candidate root since it has in-degree=0 while the second is a dead-end since it has out-degree=0.
  \textbf{B)}
  In the biedged graph the strongly connected components SCC are computed and the condensation graph is build having each SCC as supernode. Then the candidate roots are the supernodes like $1D\_L$ and other normal nodes like $(i+1)\_L$ both having in-degree=0. 
  \textbf{C)}
  An artificial black edge with nodes $00\_L$,$00\_R$ is added and $00\_R$ is connected with grey edges to every candidate root of the condensation graph with the procees expalined about the respective SCC of the supernode of the condensation graph. That makes $00\_L$ a working root for the graph.Note that by this process the snarl set of the graph might alter, like snarl ($i\_R$, $(i+4)\_L$), is not any more a snarl, but all ultrabubbles, like ($1\_R$, $4\_L$), and their frontier nodes remain unaffected.
  }
  \label{fig:dagroot}
\end{figure}

The conceptual strategy for creating a single source in a general graph $G$ is to transform the graph into a DAG $G'$, identify all $N^0$ nodes with in-degree zero in $G'$, and add a single node with outgoing edges to each node in $N^0$. Constructing a single sink requires a similar approach, except that the new sink node is connected with incoming edges from each of the nodes with out-degree zero in $G'$. Importantly, for our purpose, each ultrabubble in $G$ must be present in $G'$; that is, if the set of ultrabubbles in a graph $G$ is $UB(G)$, then $UB(G) \subseteq UB(G')$.

The following steps describe the details of this construction algorithm.
\begin{enumerate}
    \item \textbf{Transforming $G$ to a DAG.} Any graph $G$ can be transformed into a DAG by contracting each strongly connected component in $G$ into a single node. The transformed graph $G_C$ is known as the condensation of $G$. Identifying all strongly connected components, contracting these and keeping a map $f$ between nodes in $G_C$ and $G$ is $O(m + n)$ using e.g. Tarjan's algorithm \cite{21}. However, as $UB(G) \supset UB(G_C)$ if at least one of the strongly connected components of $G$ contains an ultrabubble, we use $G_C$ as an intermediary to identify all strongly connected components in $G$ that will be connected to the source in the transformed graph (see Fig.~\ref{fig:dagroot}\textbf{A}).

    \item \textbf{Identifying all nodes $N^0_C$ with in-degree zero in $G_C$.} This simple operation is $O(n + m)$ and identifies all strongly connected components that must be connected to the source (see Fig.~\ref{fig:dagroot}\textbf{B}). 

    \item \textbf{Connecting each strongly connected component in $N^0_C$ to the source.}
    Note that for each node in $N^0_C$ that maps to a single node $v \in G$, $v$ can be directly connected to the source by adding an outgoing grey edge from the source to $v$. The strongly connected components represented by the remaining nodes $v \in N^0_C$ are each connected using the following procedure (see Fig.~\ref{fig:dagroot}\textbf{C} and Fig.~\ref{fig:mroots}):
    \begin{enumerate}
    \item If the strongly connected component mapped by $v$ (that is, $f(v)$) contains no $R-L$ snarls (and therefore no ultrabubbles), the source can be connected to any (left) node in $f(v)$, which is $O(1)$.
    \item Otherwise, identify the set $S_{RL}$ of non-nested $R-L$ snarls of $f(v)$ and select the black edge $e$ connecting the right node $r$ of an arbitrary snarl in $S_{RL}$ and some node $n$ in $f(v)$. Create two new nodes $r'$ and $n'$ and split $e$ into $e_l$ and $e_r$ such that $e_l$ connects $n$ and $r'$ and $e_r$ connects $n'$ and $r$. Label the $e_l$ and $e_r$ with the label of $e$ and connect the source to $n'$. The latter operation of creating $r'$ and $n'$, splitting $e$, and connecting the source is $O(1)$.
    \end{enumerate}
\end{enumerate}

\begin{figure}[H]
  \centering
  \includegraphics[width=\linewidth]{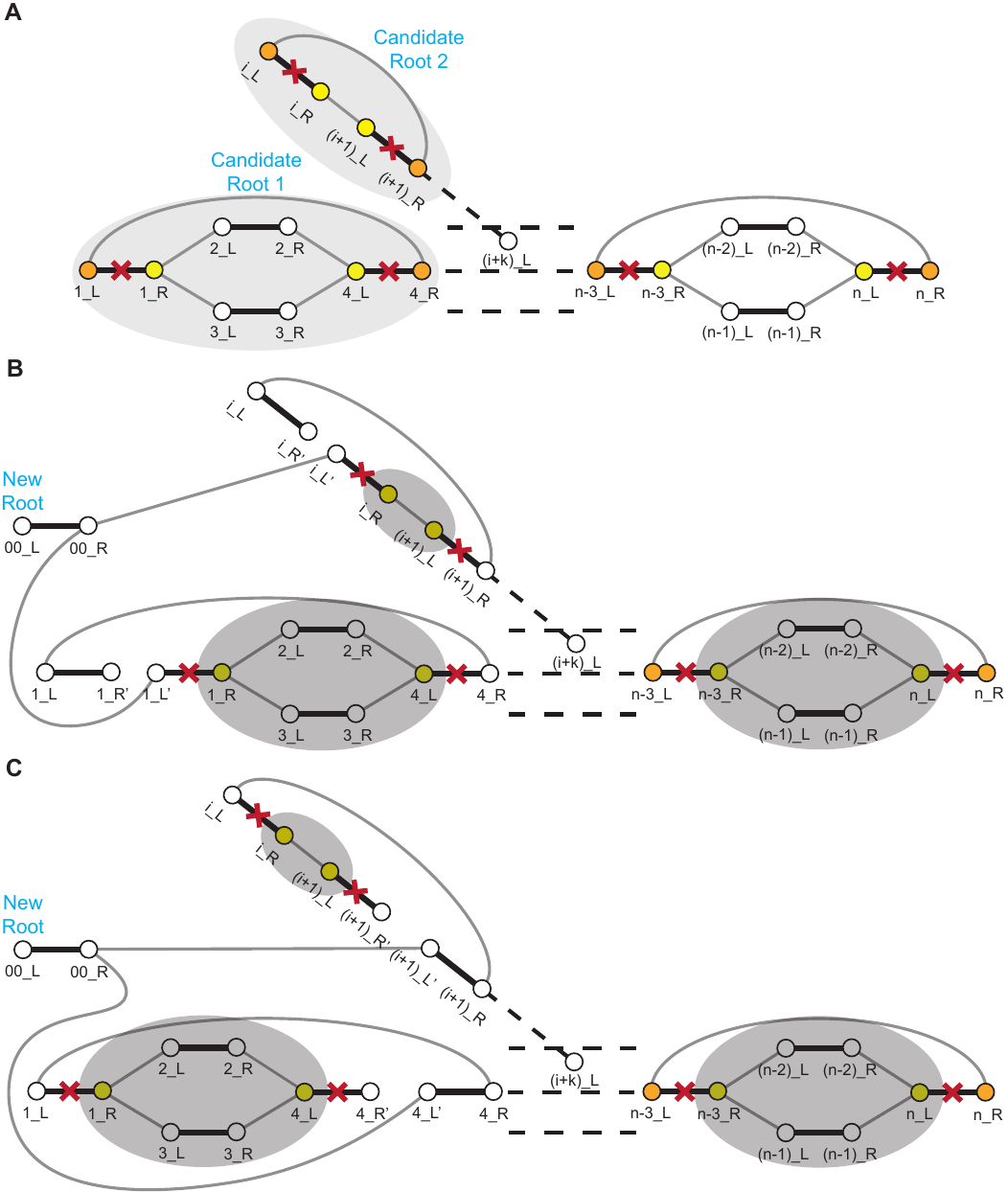} 
  \caption{
  \textbf{A)}A graph without any candidate roots has at least one SCC. Here we have 3 SCC and 2 of them are candidate roots.
  \textbf{B)} Following our approach to add an artificial root, here black edges ($i\_L$,$i\_R$) and ($1\_L$,$1\_R$) were splitted. Note that all ultrabubbles remain unaffected after the artificial root was added.
  \textbf{C)}Following our approach to add an artificial root, here black edges ($(i+1)\_L$,$(i+1)\_R$) and ($4\_L$,$4\_R$) were splitted. Note that again all ultrabubbles remain unaffected after the artificial root was added.
  }
  \label{fig:mroots}
\end{figure}

One approach to identify whether an $R-L$ snarl $(r_i, l_i)$ is nested within another $R-L$ snarl $(r_j, l_j), i \ne j$ is to investigate the finishing times $t_f$ of $r_j, r_i, l_i, l_j$ in a depth first search of $f(v)$ starting in $r_j$. First, by the definition of an $R-L$ snarl, for any $R-L$ snarl $(r_j, l_j)$ and a depth first search starting in $r_j$, $l_j$ is a proper descendant of $r_j$ in the corresponding depth first search forest and, consequently, $t_f(r_j) > t_f(l_j)$. Moreover, for any $R-L$ snarl $(r_j, l_j)$ in a strongly connected component $f(v)$ and for any depth first search starting outside the sub-graph defined by $(r_j, l_j)$, $t_f(r_j) > t_f(l_j)$. In contrast, starting the depth first search anywhere within the sub-graph defined by $(r_j, l_j) \setminus r_j$ (that is, excluding $r_j$ but including $l_j$), gives $t_f(r_j) < t_f(l_j)$. Consequently, for any $R-L$ snarl $(r_i, l_i)$ nested within another $R-L$ snarl $(r_j, l_j), i \ne j$ in a strongly connected component, a depth first search starting in $r_i$ will have $t_f(r_j) < t_f(l_j)$. Conversely, for any snarl $R-L$ snarl $(r_i, l_i)$ nested within another $R-L$ snarl $(r_j, l_j), i \ne j$ and a depth first search starting in $r_j$, $r_i$ and $l_i$ are proper descendants of $r_j$ in the corresponding depth first search forest and $t_f(r_j) > t_f(r_i) > t_f(l_i)$. More specifically, in the case where these nested snarls are in a strongly connected component, for any depth first search starting outside the sub-graph defined by the depth first search path from $r_j$ to $r_i$, $f_t(r_j) > f_t(r_i)$. Consequently, for nested snarls in a strongly connected component $f(v)$, starting a depth first search at a node $r_i$ of any $R-L$ snarl $(r_i, l_i) \in f(v)$ can result in two possibilities: (i) for all $R-L$ snarls $(r_j, l_j) \in f(v), t_f(r_j) > t_f(l_j)$, in which case the snarl $(r_i, l_i)$ is not nested within another $R-L$ snarl; (ii) for some $R-L$ snarl $(r_j, l_j) \in f(v), t_f(r_j) < t_f(l_j)$, in which case the snarl $(r_i, l_i)$ is nested within the $R-L$ snarl $(r_j, l_j)$. In case (i), we have identified a non-nested snarl and can proceed with the construction algorithm. In case (ii), $(r_i, l_i)$ can be nested within multiple $R-L$ snarls, all of which will satisfy $t_f(r_j) < t_f(l_j)$. Assuming that the set of such nested $R-L$ snarls is $S_N$ (that is, $\forall (r_i, l_i) \in S_N, t_f(r_i) < t_f(l_i)$), the $R$ node $r_i$ of the outermost snarl in $S_N$ will have the highest finishing time. Identifying the set $S_N$ is $O(V_{f(v)})$, where $V_{f(v)}$ is the number of nodes in $f(v)$; the same is true for identifying the $R$ node in $S_N$ with the highest finishing time, as the finishing times for the $R$ nodes in $S_N$ are sorted. Doing the initial depth first search is $O(E_{f(v)}+V_{f(v)})$, where $E_{f(v)}$ is the number of edges in $f(v)$, which gives $O(n+m)$ for the whole graph.

\begin{figure}[H]
  \centering
  \includegraphics[width=\linewidth]{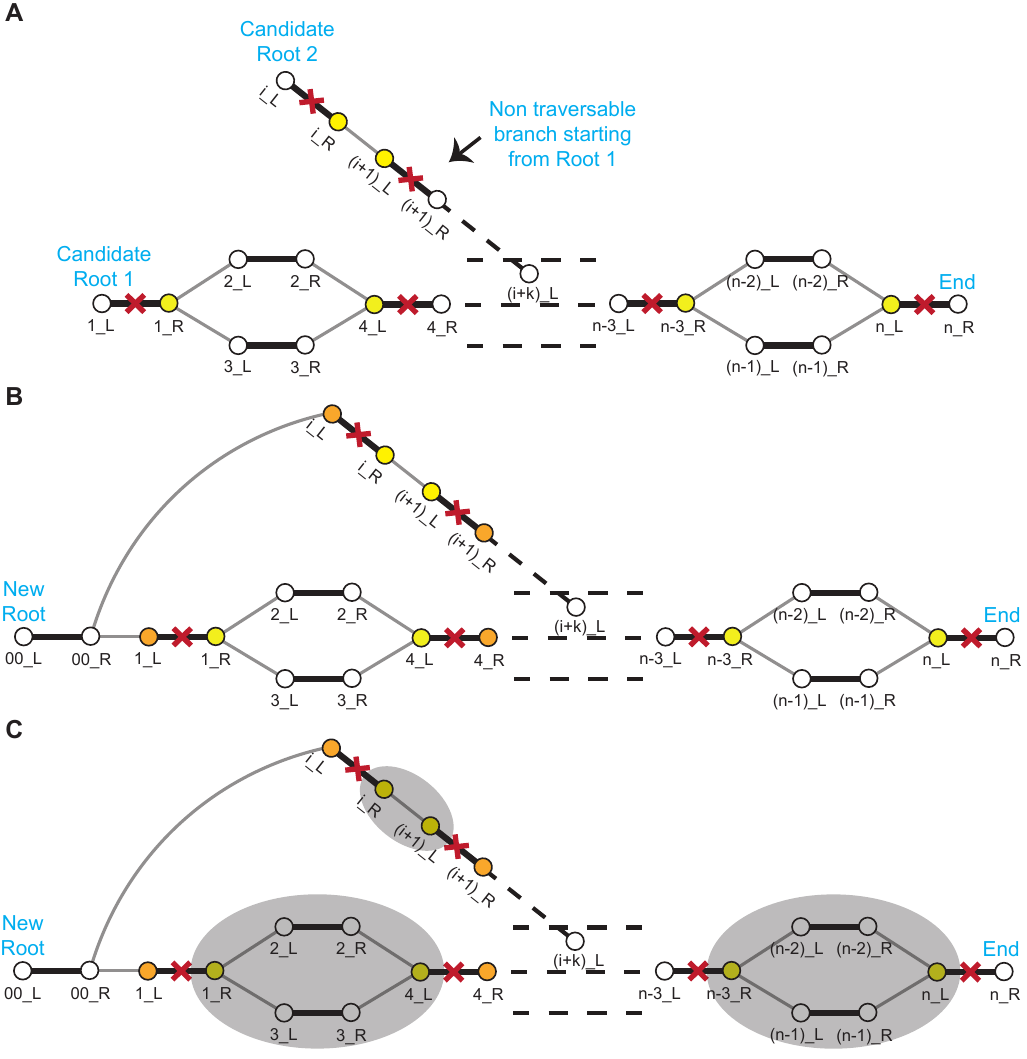} 
  \caption{
  \textbf{A)} A biedged graph that has more than one candidate roots, cannot be fully traversed by any of them since always there will be a non traversable branch.
  \textbf{B)} Adding an artificial black edge with nodes $00\_L$, $00\_R$ and then connecting $00\_R$ with grey edges to every candidate root of the biedged graph we end up having $00\_L$ as a working root.
  \textbf{C)} After the addition of the artificila root, while the snarl set is altered, all ultrabubbles ($1\_R$, $4\_L$), ($i\_R$, $(i+1)\_L$) and ($(n-3)\_R$, $n\_L$) remain unaffected.  
  }
  \label{fig:msroots}
\end{figure}

\subsection{Computing the cycle closing nodes and the corresponding grey cycle closing edges}\label{sec:cycleclose}
A cycle in a directed graph can be easily detected by a depth-first search (DFS). In a standard DFS, nodes are classified as unvisited, on the recursion stack, or fully explored (often denoted WHITE, GRAY and BLACK) and Tarjan showed \cite{21} during a DFS any back edge (an edge from a node to one of its GRAY ancestors on the recursion stack) signals the presence of a cycle. Later, Tarjan also showed \cite{22} that the above approach holds true even in graphs that they have nested and overlapping cycles since he proved that every directed cycle includes at least one back edge. Johnson’s algorithm \cite{20}, fully enumerates all simple-elementary cycles in a directed graph but  its time complexity $O((V+E)(C+1))$ is proportional to the number of cycles $C$, which may be exponential to the number of nodes of the graph.
However by our approach and because of Lemma~\ref{lem:3.6}, we are interested only at the cycle closing nodes and their corresponding (grey) closing edges and not to the explicit cycle paths. Thus, we use Tarjan approach of \cite{21} following the WHITE, GREY, BLACK edges and the completeness of this approach by \cite{22} and therefore in linear time $O(V+E)$ of a DFS execution we can mark every cycle closing node and corresponding cycle closing grey edge of every cycle of the given biedged graph including these in nested cycles.

\subsection{Answering LCA queries in linear time}\label{sec:lcalinear}
In general directed graphs, LCA queries can be handled efficiently without enumerating all paths as we mentioned in Section 3.2. We first contract all strongly connected components (SCCs) to build the condensation DAG, which is acyclic and preserves reachability \cite{23}. On this DAG, a lowest common ancestor of two nodes is any minimal common ancestor; that is a common ancestor with no descendant that also reaches both, and such LCAs may form a set and not a single node \cite{24}. If no common ancestor exists, then the LCA set is empty. To return one node per query, we can choose a deterministic representative from one of the minimal common-ancestor components defining tie breaks. For few queries, two reverse searches on the condensation DAG and an intersection are enough while for many queries, precomputing ancestor information (sets or dense bitsets) yields fast lookups. For all-pairs approach with $O(1)$ queries, there are DAG-specific algorithms with heavy preprocessing: a classic subcubic framework \cite{24}, an $O(nm)$ all-pairs method effective on sparse DAGs \cite{31}, and the current best sub cubic bound $O(n^{2.447})$  \cite{32} which is harder to implement. Here n and m refer to the number of vertices and edges of the condensation DAG.

For finding the LCA in a tree, we consider three classic approaches. The nearest common ancestor approach of Harel and Tarjan \cite{23} which supports online LCA queries in $O(1)$ time after $O(n)$ preprocessing, but it is relatively hard to implement.
A simpler alternative is binary lifting (jump pointers), which preprocesses $2^i$$-$ancestors in $O(n \log n)$ time and answers each LCA query in $O(n \log n)$
 time by first equalizing depths and then lifting both nodes simultaneously using decreasing powers of two, based on the jump-pointer table in \cite{24}.
The Euler tour +RMQ approach with ±1 depth array \cite{25} which reduces LCA to RMQ on an Euler depth sequence with adjacent differences ±1 and has a preprocessing $O(n)$ and query $O(1)$ complexity.
Although Harel-Tarjan’s algorithm and the Euler-tour + RMQ approach share the same asymptotic complexity, we chose the Euler-tour + RMQ method because it has a simpler and more straightforward implementation. We did not implement the space-efficient RMQ method of \cite{26} which provides a more space efficient variant of the same method, since this would add implementation complexity while our concern is the running time goal.

\bibliography{sn-bibliography}

\end{document}